\documentclass[11pt,a4paper]{amsart}
\usepackage[mathlines,displaymath]{lineno}
\usepackage{amssymb}
\usepackage[dvips,bookmarks,colorlinks=true,linkcolor=blue,urlcolor=red]{hyperref}
\numberwithin{equation}{section}
\newcommand{\D}{\displaystyle}

\newcommand{\car}{\mathbf{1}}
\newcommand{\R}{{\mathbb R}}

\newcommand{\Z}{{\mathbb Z}}
\newcommand{\N}{{\mathbb N}}
\newcommand{\C}{{\mathbb C}}

\newcommand{\esp}{{\mathbb E}}
\newcommand{\pro}{{\mathbb P}}

\newcommand{\supp}{{\rm supp}\,}

\newcommand{\tr}{{\rm tr}\,}

\newcommand{\dist}{{\rm dist}\,}

\newcommand{\vers}{\operatornamewithlimits{\to}}
\DeclareMathOperator*{\esssup}{ess\,sup}
\DeclareMathOperator*{\essinf}{ess\,inf}
\theoremstyle{plain}
\newtheorem{Th}{Theorem}[section]
\newtheorem{Le}{Lemma}[section]

\theoremstyle{definition}
\newtheorem{Rem}{Remark}[section]
\title{Spectral statistics for weakly correlated random potentials}
\author{Fr{\'e}d{\'e}ric Klopp} \address[Fr{\'e}d{\'e}ric Klopp]{IMJ, UMR CNRS 7586
  Universit{\'e} Pierre et Marie Curie Case 186, 4 place Jussieu F-75252
  Paris cedex 05 France}
\email{\href{mailto:klopp@math.jussieu.fr}{klopp@math.jussieu.fr}}
\begin{document}
%
\begin{abstract}
  We study localization and derive stochastic estimates (in
  particular, Wegner and Minami estimates) for the eigenvalues of
  weakly correlated random discrete Schr{\"o}dinger operators in the
  localized phase. We apply these results to obtain spectral
  statistics for general discrete alloy type models where the single
  site perturbation is neither of finite rank nor of fixed sign. In
  particular, for the models under study, the random potential exhibits
  correlations at any range.
  \vskip.5cm\noindent \textsc{R{\'e}sum{\'e}.}
  On {\'e}tudie la localisation et on obtient des estim{\'e}es probabilistes
  (en particulier des estim{\'e}es de Wegner et de Minami) pour les
  valeurs propres d'op{\'e}rateurs de Schr{\"o}dinger al{\'e}atoires discrets
  faiblement corr{\'e}l{\'e}s. Les r{\'e}sultats obtenus sont ensuite appliqu{\'e}s
  pour {\'e}tudier les statistiques spectrales de mod{\`e}les g{\'e}n{\'e}raux de type
  alliage pour lesquels le potentiel de simple site n'est ni de signe
  constant ni de support compact. En particulier, pour les mod{\`e}les
  {\'e}tudi{\'e}s, le potentiel al{\'e}atoire est corr{\'e}l{\'e} {\`a} toutes les {\'e}chelles.
\end{abstract}
\maketitle
\section{Introduction: a model with long range correlations}
\label{sec:intr-model-with}
Consider a single site potential $u:\ \Z^d\to\R$. Assume that
\begin{description}
\item[(S)] $u\in\ell^1(\Z^d)$;
\item[(H)] the continuous function $\D \theta\mapsto\sum_{n\in\Z^d}
  u_ne^{in\theta}$ does not vanish on $\R^d$.
\end{description}
Let $\omega=(\omega_n)_{n\in\Z^d}$ be real valued bounded
i.i.d. random variables.\\
Define the random ergodic Schr{\"o}dinger operator $H_\omega$ as
\begin{equation}
  \label{eq:13}
  H_\omega=-\Delta+\lambda\sum_{n\in\Z^d}\omega_n \tau_n u
\end{equation}
where
\begin{itemize}
\item $\tau_n u$ is the potential $u$ shifted to site $n$
  i.e. $(\tau_n u)_k=u(k-n)$
\item the coupling constant $\lambda$ is positive.
\end{itemize}
Define $\mu$ to be the common distribution of $(\omega_n)_{n\in\Z^d}$
and let $S$ be the concentration function of $\mu$, that is, $\D
S(s)=\sup_{a\in\R}\mu([a,a+s])$ for $s\geq0$.\\
Let us now assume that
\begin{description}
\item[(R)] $S$ is Lipschitz continuous at $0$ i.e. there exists $C>0$
  such that $S(s)\leq C s$ for $s\geq0$.
\end{description}
For $\Lambda\subset\Z^d$ a finite cube and $J\subset\R$ a measurable
set, we define $P_{\omega}^{(\Lambda)}(J)$ to be the spectral
projector of $H_\omega(\Lambda)$ that is $H_{\omega}$ restricted to
$\Lambda$ (with periodic boundary conditions) onto the energy interval
$J$
\begin{equation*}
  P_{\omega}^{(\Lambda)}(J)=\car_J(H_\omega(\Lambda)).
\end{equation*}
Define $N$ the integrated density of states of $H_\omega$ as
\begin{equation}
  \label{eq:2}
  N(E)=\lim_{|\Lambda|\to+\infty}\frac{1}{|\Lambda|}\tr
  P_{\omega}^{(\Lambda)}((-\infty,E])
\end{equation}
where $\tr$ denotes the trace.\\
Almost surely, the limit~\eqref{eq:2} exists for all $E$ real (see
e.g.~\cite{MR94h:47068,MR2509110}); it defines the distribution
function of some probability measure, say, $dN(E)$, the support of
which is the almost sure spectrum of $H_\omega$ (see
e.g.~\cite{MR94h:47068,MR2509110}). Let $\Sigma$ denote the almost
sure spectrum of $H_\omega$.\\
We first prove
\begin{Th}
  \label{thr:3}
  Under the assumptions (S), (H) and (R), for $H_\omega$, one has
  \begin{enumerate}
  \item the integrated density of states $N$ is uniformly Lipschitz
    continuous; $N$ is almost everywhere differentiable
  \item $H_\omega$ satisfies a Wegner and a Minami estimate, that is,
    there exists $C_\lambda>0$ such that, for any bounded interval $I$
    and $k$, a positive integer, we have
      \begin{equation*}
        \begin{aligned}
          \esp\left[\tr(P_\omega^{(\Lambda)}(I))\cdot\left(
              \tr(P_\omega^{(\Lambda)}(I)-1\right)\cdots\left(
              \tr(P_\omega^{(\Lambda)}(I)-(k-1)\right) \right]&\leq
          \left(C_\lambda\,|I|\,|\Lambda|\right)^k;
        \end{aligned}
      \end{equation*}
    \item for $\lambda$ sufficiently large, the whole spectrum of
      $H_\omega$ is localized i.e. there exists $\eta$ such that, for
      any $L\geq1$, if $\Lambda=\Lambda_L$ is the cube of center $0$
      and side-length $L$, one has, for any $L\geq1$ and any $p>d$,
      there is $q=q_{p,d}$ so that, for any $L$ large enough, the
      following holds with probability at least $1 - L^{-p}$: for any
      eigenvector $\varphi_{\omega,\Lambda,j}$ of $H_\omega(\Lambda)$,
      there exists a center of localization $x_{\omega,\Lambda,j}$ in
      $\Lambda$, so that for any $x\in \Lambda$, one has
      \begin{equation*}
        \|  \varphi_{\omega,\Lambda,j} \|_x \leq L^q
        e^{-\eta|x-x_{\omega,\Lambda,j}|}. 
      \end{equation*}
  \end{enumerate}
\end{Th}
\noindent As we shall see below, these conclusion holds for more
general models of random Schr{\"o}dinger operators with weakly dependent
randomness (see Theorems~\ref{thr:1},~\ref{thr:5} and~\ref{thr:6}
below).
\par For $L\in\N$, let $\Lambda=\Lambda_L=[-L,L]^d\cap\Z^d\subset\Z^d$
be a large box; recall that $H_\omega(\Lambda)$ is the operator
$H_\omega$ restricted to $\Lambda$ with periodic boundary
conditions. Let $|\Lambda|$ be the volume of $\Lambda$
i.e. $|\Lambda|=(2L+1)^d$.
\par $H_\omega(\Lambda)$ is an $|\Lambda|\times |\Lambda|$ real symmetric
matrix. By point (2) of Theorem~\ref{thr:3}, we know that the
eigenvalues of $H_\omega(\Lambda)$ are almost surely simple. Let us
denote them ordered increasingly by $E_1(\omega,\Lambda)<
E_2(\omega,\Lambda)< \cdots< E_{|\Lambda|}(\omega,\Lambda)$.
\par Let $E_0$ be an energy in $\Sigma$ such that $N$ is
differentiable at $E_0$. Let $n(E_0)$ denote the derivative of $N$ at
the energy $E_0$. The local level statistics near $E_0$ is the point
process defined by
\begin{equation}
  \label{eq:14}
  \Xi(\xi,E_0,\omega,\Lambda) = 
  \sum_{j=1}^{|\Lambda|} \delta_{\xi_j(E_0,\omega,\Lambda)}(\xi)
\end{equation}
where
\begin{equation}
  \label{eq:15}
  \xi_j(E_0,\omega,\Lambda)=|\Lambda|\,n(E_0)\,(E_j(\omega,\Lambda)-
  E_0),\quad 1\leq j\leq |\Lambda|.
\end{equation}
We prove
\begin{Th}
  \label{thr:4}
  In addition to (S), (H) and (R), assume that
  \begin{description}
  \item[(D)] for some $\D\eta>d-\frac12$,
    \begin{equation*}
      \limsup_{|n|\to+\infty}|n|^\alpha\,|u(n)|<+\infty.
    \end{equation*}
  \end{description}
  Then, for $\lambda$ sufficiently large, for $E_0$, an energy in
  $\Sigma$ such that $n(E_0)$ exist and be positive, when
  $|\Lambda|\to+\infty$, the point process $\Xi(E_0,\omega,\Lambda)$
  converges weakly to a Poisson process on $\R$ with intensity
  $1$. That is, for $p>0$ arbitrary, for arbitrary non empty open two
  by two disjoint intervals $I_1,\dots,I_p$ and arbitrary integers
  $k_1,\cdots,k_p$, one has
  \begin{multline}
    \label{eq:19}
    \lim_{|\Lambda|\to+\infty}
    \pro\left(\left\{\omega;\
        \begin{aligned}
          &\#\{j;\ \xi_j(E_0,\omega,\Lambda)\in
          I_1\}=k_1\\&\vdots\hskip3cm\vdots\\
          &\#\{j;\ \xi_j(E_0,\omega,\Lambda)\in I_p\}=k_p
        \end{aligned}
      \right\}\right)\\=\frac{|I_1|^{k_1}}{k_1!}e^{-|I_1|}\cdots
    \frac{|I_p|^{k_p}}{k_p!}e^{-|I_p|}.
  \end{multline}
\end{Th}
\noindent When the single site potential $u$ is supported in a single
point, this result was first obtained~\cite{MR97d:82046}. This result
is typical of the localized phase of random operators; a general
analysis of the spectral statistics for the eigenvalues was developed
in~\cite{Ge-Kl:10} (see also~\cite{MR2885251}). There, it was shown
that, in the localized regime, under the additional ``independence at
a distance'' (IAD) assumption, if the random model satisfies a Wegner
and a Minami estimate (see below for more details), then the
convergence to Poisson for the local statistics holds.\\
The (IAD) assumption requires that there exists some fixed positive
distance, say, $D$ such that, for any $\Lambda$ and $\Lambda'$ at
least at a distance $D$ apart from each other, the random operators
$H_\omega(\Lambda)$ and $H_\omega(\Lambda')$ are independent. For a
general $u$ chosen as above, this assumption is not fulfilled.  There
are long range correlations. We will show that they do not suffice to
induce correlations between the eigenvalues asymptotically.
\begin{Rem}
  \label{rem:1}
  Note that if $u$ is not too ``lacunary'' assumption (S) actually
  implies that $|u(n)|\lesssim |n|^{-d}$, that is, in particular that
  assumption (D) holds.\\
  From the analysis done in~\cite{Ge-Kl:10}, it is quite clear that,
  for the model~\eqref{eq:13}, if $|u(n)|=o(|n|^{-2d+1-\varepsilon})$
  for some $\varepsilon>0$, then, the correlation are much smaller
  than the typical spacing between the eigenvalues; thus, they should
  not play a role in the statistics.\\
  As the proof of Theorem~\ref{thr:4} shows (see, in particular, the
  proof of Lemma~\ref{le:3}), condition (D) can be relaxed to 
  \begin{equation*}
    \limsup_{|n|\to+\infty}|n|^{d-1/2}\,\log|n|\,|u(n)|<\eta
  \end{equation*}
  for sufficiently small $\eta$.
\end{Rem}
\noindent Following the proof of Theorem~\ref{thr:4} and that
of~\cite[Theorem 1.13]{Ge-Kl:10}, one proves
\begin{Th}
  \label{thr:8}
  Under the assumptions of Theorem~\ref{thr:4}, for $1\leq j\leq
  |\Lambda|$, let $x_j(E_0,\omega,\Lambda)$ be a localization center
  associated to the eigenvalue $E_j(\omega,\Lambda)$ (see
  Theorem~\ref{thr:3}). Then, the joint (eigenvalue - localization
  center) local point process defined by
  \begin{equation*}
    \Xi^2_\Lambda(\xi,x;E_0,\Lambda) = \sum_{j=1}^N
    \delta_{\xi_j(E_0,\omega,\Lambda)}(\xi)
    \otimes\delta_{x_j(E_0\omega,\Lambda)/L}(x)
  \end{equation*}
  converges weakly to a Poisson point process on $\R\times [-1/2,1/2]^d$
  with intensity $1$.
\end{Th}
\noindent In~\cite{Ge-Kl:10}, many other spectral statistics are
studied; using the ideas developed there and in the present work, they
can also be studied for the model~\eqref{eq:13}. \vskip.2cm\noindent
To close this introduction, we note that, in~\cite{Ve-Ta:12},
Theorem~\ref{thr:4} is proved for $H_\omega$ under more restrictive
assumptions: $u$ is assumed to be of compact support and the common
density of the random variables $(\omega_n)_{n\in\Z^d}$ is supposed to
be sufficiently regular. In particular, under these assumptions, the
operator $H_\omega$ satisfies the independence at a distance
assumption.
\section{Weakly correlated random potentials}
\label{sec:weakly-corr-rand}
On $\ell^2(\Z^d)$, consider the discrete Anderson model
\begin{equation}
  \label{eq:1}
  H_{\tilde\omega}=-\Delta+\lambda V_{\tilde\omega}\text{ where
  }V_{\tilde\omega}=((\tilde\omega_n\delta_{nm}))_{(n,m)\in\Z^d\times\Z^d} 
\end{equation}
and $\delta_{nm}$ is the Kronecker symbol and $\lambda$ is positive
coupling constant.\\
Assume that the random variables $(\tilde\omega_n)_{n\in\Z^d}$ are non
trivial, real valued and bounded. They are not assumed to be
independent or identically distributed.\\
Let $\Lambda\subset\Z^d$ be a finite cube and $J\subset\R$ a
measurable set, we define $P_{\tilde\omega}^{(\Lambda)}(J)$ to be the
spectral projector of $H_{\tilde\omega}(\Lambda)$, that is,
$H_{\tilde\omega}$ restricted to $\Lambda$ (with periodic boundary
conditions) onto the energy interval $J$
\begin{equation*}
  P_{\tilde\omega}^{(\Lambda)}(J)=\car_J(H_{\tilde\omega}(\Lambda)).
\end{equation*}
For the operator $H_{\tilde\omega}$, we will now state a
number of results on
\begin{itemize}
\item eigenvalue decorrelation estimates,
\item localization,
\item local representations of the eigenvalues of
  $H_{\tilde\omega}(\Lambda)$ in some energy interval.
\end{itemize}
In the case of independent random variables, these results are the
basic building blocks of the analysis of the spectral statistics
performed in~\cite{Ge-Kl:10}. We shall then show that these results
can be applied to the model~\eqref{eq:13} and analyze the spectral
statistics of this model to prove Theorem~\ref{thr:4}.
\subsection{Basic estimates on eigenvalues - decorrelation estimates}
\label{sec:basic-estim-eigenv}
We first need some estimates on the occurrence of eigenvalues in given
intervals as well as on their correlations. These give descriptions of
the eigenvalues seen as random variables that have proved crucial in
the study of spectral statistics. In these estimates, we actually do
not need the random variables $(\tilde\omega_n)_{n\in\Z^d}$ to be
identically distributed.\\ 
Define $\tilde\mu_m$ to be the distribution of
$\tilde\omega_m$ conditioned on $(\tilde\omega_n)_{n\not=m}$ and let
$\tilde S_m$ be the concentration function of $\tilde\mu_m$, that is,
$\D \tilde S_m(s)=\sup_{a\in\R} \tilde\mu_m([a,a+s])$. For
$\Lambda\subset\Z^d$, define
\begin{equation*}
  \tilde
  S_\Lambda:=\sup_{m\in\Lambda}\,\esssup_{(\tilde\omega_n)_{n\not=m}}
  \tilde S_m  
\end{equation*}
We prove
\begin{Th}
  \label{thr:1}
  There exists $C>0$ such that
  \begin{itemize}
  \item for any bounded interval $I$, we have
    \begin{equation}
      \label{eq:7}
      \esp\left[\tr(P_{\tilde\omega}^{(\Lambda)}(I))
      \right] \leq C\tilde S_{\Lambda}(|I|)|\Lambda|;
    \end{equation}
  \item for any two bounded intervals $I_1\subset I_2$, we have
    \begin{equation}
      \label{eq:3}
      \esp\left[\tr(P_{\tilde\omega}^{(\Lambda)}(I_1))
        \left(\tr(P_{\tilde\omega}^{(\Lambda)}(I_2)-1\right)
      \right]\leq C^2 \tilde S_{\Lambda}(|I_1|)\tilde
      S_{\Lambda}(|I_2|)|\Lambda|^2;
    \end{equation}
  \item specializing to the case $I_1=I_2=I$,~(\ref{eq:3}) reads
    \begin{equation}
      \label{eq:4}
      \esp\left[\tr(P_{\tilde\omega}^{(\Lambda)}(I))
        \left(\tr(P_{\tilde\omega}^{(\Lambda)}(I)-1\right)
      \right]\leq C^2 \left(\tilde S_{\Lambda}(|I|)|\Lambda|\right)^2;
    \end{equation}
  \item     for any $k\geq1$, and arbitrary intervals $I_1\subset
    I_2\subset\cdots\subset I_k$, one has
    \begin{multline}
      \label{eq:5}
      \esp\left[\tr(P_{\tilde\omega}^{(\Lambda)}(I_1))
        \left(\tr(P_{\tilde\omega}^{(\Lambda)}(I_2)-1\right)\cdots
        \left(\tr(P_{\tilde\omega}^{(\Lambda)}(I_k)-(k-1)\right)
      \right]\\\leq C^k |\Lambda|^k\prod_{j=1}^k\tilde
      S_{\Lambda}(|I_j|);
    \end{multline}
  \item specializing to the case $I_1=\cdots=I_k=I$,~(\ref{eq:5})
    reads
    \begin{multline}
      \label{eq:6}
      \esp\left[\tr(P_{\tilde\omega}^{(\Lambda)}(I))
        \left(\tr(P_{\tilde\omega}^{(\Lambda)}(I)-1\right)\cdots
        \left(\tr(P_{\tilde\omega}^{(\Lambda)}(I)-(k-1)\right)
      \right]\\\leq C^k \left(\tilde
        S_{\Lambda}(|I|)|\Lambda|\right)^k.
    \end{multline}
  \end{itemize}
\end{Th}
\begin{Rem}
  \label{rem:2}
  In Theorem~\ref{thr:1}, as the proof shows, the results stay valid
  if the spectral projector $P_{\tilde\omega}^{(\Lambda)}(I)$ is
  replaced by the spectral projector on $I$ for $H_{\tilde\omega}$
  restricted to $\Lambda$ with different boundary conditions.
\end{Rem}
\noindent Estimate~\eqref{eq:7} is the so called ``Wegner estimate'';
it was first obtained in~\cite{MR639135} when the random variables
$(\omega_n)_{n\in\Z^d}$ are independent and identically
distributed. We refer the reader to~\cite{MR2378428} for more details
on Wegner estimates for various models. Estimate~(\ref{eq:5}) is the
so called ``Minami estimate''; it was first obtained
in~\cite{MR97d:82046} when $(\tilde\omega_n)_n$ are independent and
identically distributed (see also~\cite{MR2290333,MR2360226}). Later,
estimates~\eqref{eq:4}~--~\eqref{eq:6} were proved for independent
random variables in~\cite{MR2505733}. The method developed
in~\cite{MR2505733} can be transposed verbatim to the present case as
the proof below shows.
\begin{proof}[Proof of Theorem~\ref{thr:1}]
  The proof essentially consists in repeating verbatim the proof
  of~\cite[Theorem 2.1]{MR2505733} replacing expectations by
  conditional expectations. We give some details for
  the reader's convenience.\\
  We first recall the following well-known spectral averaging lemma
  (in a form specialized to our setting)
  \begin{Le}
    \label{le:7}
    Let $H_0$ be a self-adjoint operator of $\ell^2(\Z^d)$ and define
    $H_t=H_0+t\pi_0$ where $t\in\R$ and
    $\pi_0=((\delta_{n0}\delta_{0m}))_{(n,m)\in\Z^d\times\Z^d}$ is the
    orthogonal projector on the unit vector
    $\delta_0=(\delta_{n0})_{n\in\Z^d}$.\\
    Let $\car_{I}(H_t)$ be the spectral projector of $H_t$ associated
    to the energy interval $I$.\\
    Let $\mu$ be an arbitrary, non trivial, compactly supported
    probability measure on $\R$. \\
    Then, for any bounded interval $I$, one has
    \begin{equation*}
      \int_\R\langle\delta_0,\car_{I}(H_t)\delta_0\rangle
      d\mu(t)\leq 8\,S_\mu(|I|).
    \end{equation*}
  \end{Le}
  \noindent A proof of this result can be found
  e.g. in~\cite{MR2362242} or in the appendix of~\cite{MR2505733}.\\
  This immediately yields~\eqref{eq:7}. Indeed, defining the vector
  $\delta_n=(\delta_{nm})_{m\in\Z^d}$ for $n\in\Z^d$, one writes
  \begin{equation}
    \label{eq:30}
    \esp\left[\tr(P_{\tilde\omega}^{(\Lambda)}(I))
    \right]=\sum_{n\in\Lambda}\esp\left[\langle\delta_n,
      \car_{I}(H_{\tilde\omega}(\Lambda))\delta_n\rangle\right].
  \end{equation}
  Now, the operator $H_{\tilde\omega}(\Lambda)$ may be written as
  $H_{\tilde\omega}(\Lambda)=H^n_{\tilde\omega}(\Lambda)+\tilde\omega_n\pi_n$
  where $\pi_n$ is the orthogonal projector on $\delta_n$. To evaluate
  the expectation in~\eqref{eq:30}, in the $n$th term, we first
  compute the expectation with respect to $\tilde\omega_n$ conditioned
  on $(\tilde\omega_m)_{m\not=n}$ and, by Lemma~\ref{le:7}, we get
  \begin{equation*}
    \esp_{\tilde\omega_n}\left[\langle\delta_n,
      \car_{I}(H_{\tilde\omega}(\Lambda))\delta_n\rangle
      \,|\,(\tilde\omega_m)_{m\not=n}\right]\leq 8\tilde S_n(|I|).
  \end{equation*}
  Plugging this into~\eqref{eq:30} and using the definition of $\tilde
  S_\Lambda$ immediately yields~\eqref{eq:7}.\vskip.1cm\noindent
  Let us now turn to the proof of~\eqref{eq:3}; we won't give a
  detailed proof of inequality~\eqref{eq:5} as it is very similar to
  that of~\eqref{eq:3}; we refer to~\cite{MR2505733} for details. \\
  We recall~\cite[Lemma 4.1]{MR2505733} specialized to our setting.
  \begin{Le}[\cite{MR2505733}]
    \label{le:8}
    Assume we are in the setting of Lemma~\ref{le:7}. Assume moreover
    that $\tr(\car_{I}(H_0))<+\infty$ for any $I\subset\R$.\\
    Then, for arbitrary $a<b$ real and $0\leq s\leq t$, we have
    \begin{equation*}
      \tr(\car_{(a,b]}(H_s))\leq 1+\tr(\car_{(a,b]}(H_t)).
    \end{equation*}
  \end{Le}
  \noindent To prove~\eqref{eq:3}, we follow the proof
  of~\cite[Theorem 2.1]{MR2505733}. Let $M=\sup_n\esssup
  \tilde\omega_n$ and $m=\inf_n\essinf \tilde\omega_n$. Pick
  $\tau_n\geq M$. Then, by Lemma~\ref{le:8}, one computes
  \begin{equation*}
    \begin{split}
      \tr(P_{\tilde\omega}^{(\Lambda)}(I_1))
      \left(\tr(P_{\tilde\omega}^{(\Lambda)}(I_2)-1\right) &=
      \sum_{n\in\Lambda}\langle\delta_n,P_{\tilde\omega}^{(\Lambda)}(I_1)
      \delta_n\rangle
      \left(\tr(P_{\tilde\omega}^{(\Lambda)}(I_2)-1\right) \\&\leq
      \sum_{n\in\Lambda}\langle\delta_n,P_{\tilde\omega}^{(\Lambda)}(I_1)
      \delta_n\rangle
      \tr\left(P_{((\tilde\omega_m)_{m\not=n},\tau_n)}^{(\Lambda)}(I_2)\right).
    \end{split}
  \end{equation*}
  Thus, taking the expectation with respect to $\tilde\omega_n$
  conditioned on $(\tilde\omega_m)_{m\not=n}$, using the spectral
  averaging lemma, Lemma~\ref{le:7}, we obtain
  \begin{multline}
    \label{eq:31}
    \esp\left(\tr(P_{\tilde\omega}^{(\Lambda)}(I_1))
      \left(\tr(P_{\tilde\omega}^{(\Lambda)}(I_2)-1\right)\,|
      (\tilde\omega_m)_{m\not=n}\right)\\\leq 8\tilde S_\Lambda(|I_1|)
    \sum_{n\in\Lambda}
    \tr\left(P_{((\tilde\omega_m)_{m\not=n},\tau_n)}^{(\Lambda)}(I_2)\right).
  \end{multline}
  for arbitrary $(\tau_n)_{n\in\Lambda}$ such that $\tau_n\geq
  M$. Thus, we can set $\tau_n=\tilde\tau_n+M-m$ where $\tilde\tau_n$
  is the random variable $\tilde\omega_n$ conditioned on
  $(\tilde\omega_m)_{m\not=n}$.\\
  We then take the expectation with respect to $\tilde\omega$ on both
  side in~\eqref{eq:31} to obtain
  \begin{equation*}
    \begin{split}
      \esp\left(\tr(P_{\tilde\omega}^{(\Lambda)}(I_1))
        \left(\tr(P_{\tilde\omega}^{(\Lambda)}(I_2)-1\right)\right)&\leq
      8\tilde S_\Lambda(|I_1|)\\&\cdot \sum_{n\in\Lambda} \esp
      \left(\tr\left(P_{((\tilde\omega_m)_{m\not=n}
            ,\tilde\omega_n+M-m)}^{(\Lambda)}(I_2)\right)\right)\\&\leq
      C \tilde S_\Lambda(|I_1|) \sum_{n\in\Lambda} \tilde
      S_\Lambda(|I_2|)|\Lambda|
    \end{split}
  \end{equation*}
  where in the last step we have used estimate~\eqref{eq:7} (for a
  different set of random variables).\\
  Thus, we obtain
  \begin{equation*}
    \esp\left(\tr(P_{\tilde\omega}^{(\Lambda)}(I_1))
      \left(\tr(P_{\tilde\omega}^{(\Lambda)}(I_2)-1\right)\right)\leq
    C \tilde S_\Lambda(|I_1|) \tilde S_\Lambda(|I_2|)|\Lambda|^2
  \end{equation*}
  that is,~\eqref{eq:3}.\\
  This completes the proof of Theorem~\ref{thr:1}. More details can be
  found in~\cite{MR2505733}.
\end{proof}
\subsection{Localization}
\label{sec:localization}
The second ingredient needed in our analysis is localization. Using
the notations above, let us assume that
\begin{description}
\item[(\~{R})] There exists $C>0$ such that, for $s\geq0$, one has
  \begin{equation*}
   \sup_{\Lambda\subset\Z^d}\tilde S_\Lambda(s)\leq Cs. 
  \end{equation*}
\end{description}
This in particular implies that, for any $m$, $\tilde\mu_m$, the
distribution of $\tilde\omega_m$ conditioned on
$(\tilde\omega_n)_{n\not=m}$ admits a density bounded by $C$.\\
Thus, we can apply the results of~\cite{MR2002h:82051}, in particular
~\cite[criterion (1.12)]{MR2002h:82051} to obtain
\begin{Th}
  \label{thr:5}
  Assume (\~R) holds. For $\lambda$ sufficiently large, there exists
  $\eta=\eta_\lambda$ such that, for any $L\geq1$, if
  $\Lambda=\Lambda_L$ is the cube of center $0$ and side-length $L$,
  one has, for any $L\geq1$
  \begin{equation}
    \label{eq:16}
    \sup_{y\in \Lambda} \esp\left\{ \sum_{x\in\Lambda} e^{\eta|x-y|}
      \sup_{\substack{\supp f \subset I \\ |f|\le 1 }} \| \chi_x
      f(H_{\tilde \omega,\Lambda}) \chi_y \|_2  \right\}<\infty.
    \end{equation}
\end{Th}
\noindent Note that assumption (\~R) guarantees that, for any
$n\in\Z^d$, the distribution of $\tilde\omega_n$ conditioned on
$(\tilde\omega_m)_{m\not=n}$ admits a density that is bounded by a
constant that is independent of $(\tilde\omega_m)_{m\not=n}$ and $n$.\\
While in~\cite{MR2002h:82051} the proofs are given in the case of
independent random variables, in the beginning of~\cite[section
1.1]{MR2002h:82051}, it is noted that, in the case of dependent random
variables, regularity of the conditional distributions (in particular
that implied by assumption (\~R)) is sufficient to perform the same
analysis.\\
As in~\cite[Theorem 6.1]{Ge-Kl:10}, we then obtain that, one has
\begin{Th}
  \label{thr:6}
  Assume (\~R) holds. For $\lambda$ sufficiently large, there exists
  $\eta=\eta_\lambda$ such that, for any $L\geq1$, if
  $\Lambda=\Lambda_L$ is the cube of center $0$ and side-length $L$,
  one has, for any $L\geq1$, for all $p>d$, there is $q=q_{p,d}$ so
  that, for any $L$ large enough, the following holds with probability
  at least $1 - L^{-p}$: for any eigenvector
  $\varphi_{\tilde\omega,\Lambda,j}$ of $H_{\tilde\omega}(\Lambda)$,
  there exists a center of localization $x_{\tilde\omega,\Lambda,j}$
  in $\Lambda$ such that, for any $x\in \Lambda$, one has
  \begin{equation}
    \label{eq:17}
    \|  \varphi_{\tilde\omega,\Lambda,j} \|_x \le L^q
    e^{-\eta|x-x_{\tilde\omega,\Lambda,j}|};
  \end{equation}
  moreover, two localization centers are at most at a distance
  $C_p\log L$ away for each other (the positive constant $C_p$ depends
  on $p$ only).
\end{Th}
\noindent Note that in the derivation of~\eqref{eq:17}
from~\eqref{eq:16} in~\cite[the proof of Theorem 6.1]{Ge-Kl:10} the
assumption (IAD) was not used. Thus, this proof can be repeated
verbatim here to derive Theorem~\ref{thr:6}.
\subsection{A representation theorem for the eigenvalues}
\label{sec:repr-theor-eigenv}
We are now in a position to prove a representation theorem for the
eigenvalues of the random operator; it is the analogue of the
representation theorem~\cite[Theorem 1.1]{Ge-Kl:10}.\\
To this effect, though it is not strictly necessary, it will be
convenient to use the density of states of the model
$H_{\tilde\omega}$. Therefore, it is convenient to assume now that the
process $(\tilde\omega_n)_{n\in\Z^d}$ is $\Z^d$-ergodic (see
e.g.~\cite{MR2509110}). For the model~\eqref{eq:13}, this assumption
is clearly satisfied.\\
Define $N$ the integrated density of states of $H_{\tilde\omega}$ as
\begin{equation*}
  N(E)=\lim_{|\Lambda|\to+\infty}\frac{1}{|\Lambda|}\tr
  P_{\tilde\omega}^{(\Lambda)}((-\infty,E])
\end{equation*}
where $\tr$ denotes the trace. The existence of the above limit is a
consequence of the ergodicity of the process
$(\tilde\omega_n)_{n\in\Z^d}$ (see e.g.~\cite{MR2509110}).\\
Almost surely, the integrated density of states exists for all $E$
real (see e.g.~\cite{MR94h:47068,MR2509110}); it defines the
distribution function of some probability measure, say, $dN(E)$, the
support of which is the almost sure spectrum of $H_{\tilde\omega}$
(see e.g.~\cite{MR94h:47068,MR2509110}). Let $\Sigma$ denote the
almost sure spectrum of $H_{\tilde\omega}$.
To start, pick $\tilde\rho$ such that
\begin{equation}
  \label{condrho}
  0\leq\tilde{\rho}<\frac1{1+d}.
\end{equation}
Pick $E_0$ and $I_\Lambda$ centered at $E_0$ such that
$N(I_\Lambda)\asymp|\Lambda|^{-\alpha}$ for
$\alpha\in(\alpha_{d,\tilde\rho},1)$ where
$\alpha_{d,\tilde\rho}$ is defined as
\begin{equation}
  \label{eq:59}
  \alpha_{d,\tilde\rho}:=(1+\tilde\rho)\frac{d+1}{d+2}.
\end{equation}
\begin{Th}
  \label{thr:7}
  Pick an energy $E_0$ such that $N$ is differentiable at $E_0$ and
  such that $\D\frac{dN}{dE}(E_0)=n(E_0)>0$. Pick $I_\Lambda$ centered
  at $E_0$ such that $N(I_\Lambda)\asymp|\Lambda|^{-\alpha}$. There
  exists $\beta>0$ and $\beta'\in(0,\beta)$ small so that
  $1+\beta<\frac{2\alpha}{1+\tilde\rho}$ and, for $\ell\asymp L^\beta$
  and $\ell'\asymp L^{\beta'}$, there exists a decomposition of
  $\Lambda$ into disjoint cubes of the form
  $\Lambda_\ell(\gamma_j):=\gamma_j+[0,\ell]^d$ satisfying:
  \begin{itemize}
  \item $\cup_j\Lambda_\ell(\gamma_j)\subset\Lambda$,
  \item $\dist (\Lambda_\ell(\gamma_j),\Lambda_\ell(\gamma_k))\ge
    \ell'$ if $j\not=k$,
  \item $\dist (\Lambda_\ell(\gamma_j),\partial\Lambda)\ge \ell'$
  \item $|\Lambda\setminus\cup_j\Lambda_\ell(\gamma_j)|\lesssim |
    \Lambda|\ell'/\ell$,
  \end{itemize}
  and such that, for $L$ sufficiently large, there exists a set of
  configurations $\mathcal{Z}_ \Lambda $ s.t.:
  \begin{itemize}
  \item $\pro(\mathcal{Z}_\Lambda)\geq 1 - |
    \Lambda|^{-(\alpha-\alpha_{d,\tilde\rho})}$,
  \item for $\tilde\omega\in\mathcal{Z}_\Lambda $, each centers of
    localization associated to $H_{\tilde\omega}(\Lambda)$ belong to some
    $\Lambda_\ell(\gamma_j)$ and each box $\Lambda_\ell(\gamma_j)$
    satisfies:
    \begin{enumerate}
    \item the Hamiltonian $H_{\tilde\omega,\Lambda_\ell(\gamma_j)}$ has at
      most one eigenvalue in $I_\Lambda $, say,
      $E_j(\tilde\omega,\Lambda_\ell(\gamma_j))$;
    \item $\Lambda_\ell(\gamma_j)$ contains at most one center of
      localization, say $x_{k_j}(\tilde\omega,\Lambda)$, of an eigenvalue of
      $H_{\tilde\omega}(\Lambda)$ in $I_\Lambda $, say
      $E_{k_j}(\tilde\omega,\Lambda)$;
    \item $\Lambda_\ell(\gamma_j)$ contains a center
      $x_{k_j}(\tilde\omega,\Lambda)$ if and only if
      $\sigma(H_{\tilde\omega,\Lambda_\ell(\gamma_j)})\cap
      I_\Lambda\not=\emptyset$; in which case, one has
      \begin{multline}
        \label{error}
          |E_{k_j}(\tilde\omega,\Lambda)-E_j(\tilde\omega,\Lambda_\ell(\gamma_j))|
          \leq e^{-\eta \ell'/2}\\
          \text{and}\quad\mathrm{dist}(x_{k_j}(\tilde\omega,\Lambda),
          \Lambda \setminus \Lambda_\ell(\gamma_j))\geq \ell'.
      \end{multline}
      where $\eta$ is given by Theorem~\ref{thr:6}.
    \end{enumerate}
  \end{itemize}
  In particular, if $\tilde\omega\in\mathcal{Z}_\Lambda$, all the eigenvalues
  of $H_{\tilde\omega}(\Lambda)$ are described by~\eqref{error}.
\end{Th}
\noindent The proof of this result is a verbatim repetition the proof
of~\cite[Theorem 1.1]{Ge-Kl:10} verbatim. The independence at a
distance is not used; only the localization and the Wegner and Minami
estimates are used.\\
Note that now the ``local'' Hamiltonians
$(H_{\tilde\omega,\Lambda_\ell(\gamma_j)})_j$ are not necessarily
stochastically independent.\\
In our application of Theorem~\ref{thr:7} to the model~\eqref{eq:13},
the choice we will make for the parameters $\alpha$, $\beta'$, $\beta$
and $\tilde\rho$ will be quite different from the one made
in~\cite{Ge-Kl:10}. While in~\cite{Ge-Kl:10} the authors wanted to
maximize the admissible size f $|I_\Lambda$ (i.e. minimize $\alpha$),
here, our primary concern will be to control the non independence of
the ``local'' Hamiltonians
$(H_{\tilde\omega,\Lambda_\ell(\gamma_j)})_j$. To control their
correlations, we will pick $\beta'$ close to $1$. Thus, $\alpha$ and
$\beta$ will also have to be close to $1$ and $\tilde\rho$ will be
close to $0$.
\subsection{The local distribution of the eigenvalues}
\label{sec:local-distr-eigenv}
The final ingredient needed for the analysis of the spectral
statistics is a precise description of the distribution of the
spectrum of the ``local'' Hamiltonians
$(H_{\tilde\omega,\Lambda_\ell(\gamma_j)})_j$ constructed in
Theorem~\ref{thr:7} within $I_\Lambda$.\\
Pick $\ell$ large and $1\ll\ell'\ll\ell$. Consider a cube $\Lambda$ of
side-length $\ell$ i.e. $\Lambda=\Lambda_{\ell}$ and an interval
$I_\Lambda=[a_\Lambda,b_\Lambda]\subset I$ (i.e. $I_\Lambda$ is
contained in the localization region). Consider the following random
variables:
\begin{itemize}
\item $X=X(\Lambda,I_\Lambda)=X(\Lambda,I_\Lambda,\ell')$ is the
  Bernoulli random variable
\begin{equation*}
  X=\car_{H_{\tilde\omega}(\Lambda)\text{ has exactly one
      eigenvalue in }I_\Lambda\text{ with localization center in }
    \Lambda_{\ell-\ell'}};
\end{equation*}
\item $\tilde E=\tilde E(\Lambda,I_\Lambda)$ is the eigenvalue of
  $H_{\tilde\omega}(\Lambda)$ in $I_\Lambda$ conditioned on $X=1$;
\item $\D\tilde\xi=\tilde\xi(\Lambda,I_\Lambda)=\frac{\tilde
    E(\Lambda,I_\Lambda)-a_\Lambda}{b_\Lambda-a_\Lambda}$.
\end{itemize}
Clearly $\tilde\xi$ is valued in $[0,1]$; let $\tilde\Xi$ be its distribution
function.\\
We now describe the distribution of these random variables as
$|\Lambda|\to+\infty$ and $|I_\Lambda|\to0$. One proves
\begin{Le}[\cite{Ge-Kl:10}]
  \label{lemasympt}
  Assume $\lambda$ is sufficiently large and pick $\eta$ as in
  Theorem~\ref{thr:6}. One has
  \begin{equation}
    \label{eq:51}
    \left|\pro(X=1)-N(I_\Lambda)|\Lambda|\right|\lesssim
    (|\Lambda||I_\Lambda|)^{1+\rho}+N(I_\Lambda)|\Lambda|\ell'\ell^{-1}
    +|\Lambda|e^{-\eta\ell'/2}
  \end{equation}
  where $N(E)$ denotes the integrated density of states of
  $H_{\tilde\omega}$.\\
  One has, for all $x,y\in[0,1]$,
  \begin{equation}
    \label{eq:50}
    \left|(\tilde\Xi(x)-\tilde\Xi(y))\,\pro(X=1)\right|\lesssim
    |x-y||I_\Lambda||\Lambda|.
  \end{equation}
  Moreover, setting $N(x,y,\Lambda):=[N(a_\Lambda+x|I_\Lambda|)-
  N(a_\Lambda+y|I_\Lambda|)]|\Lambda|$, one has
  \begin{multline}
    \label{eq:52}
    \left|(\tilde\Xi(x)-\tilde\Xi(y))\,\pro(X=1)-N(x,y,\Lambda)\right|
    \\\lesssim (|\Lambda||I_\Lambda|)^{1+\rho}
    +|N(x,y,\Lambda)|\,\ell'\ell^{-1}+ |\Lambda|e^{-\eta\ell'/2}.
  \end{multline}
\end{Le}
\noindent The proof of this result in~\cite[Lemma 2.2]{Ge-Kl:10}
relies solely on the localization estimates obtained in
Theorem~\ref{thr:6}, the Wegner estimate~\eqref{eq:7} and the Minami
estimate~\eqref{eq:4}. Thus, it can be repeated verbatim in the
present setting for the model $H_{\tilde\omega}$.
\section{The long range correlation model~\eqref{eq:13} seen as weakly
  dependent model of the type~\eqref{eq:1}}
\label{sec:proof}
We shall now show that the long range correlation model~\eqref{eq:13}
satisfies the assumptions of
Theorems~\ref{thr:1},~\ref{thr:5},~\ref{thr:6} and~\ref{thr:7} and
Lemma~\ref{lemasympt}. The only non trivial assumption to verify is
(\~ R).
\subsection{A linear change of random variables}
\label{sec:linear-change-random}
Consider the Banach space $\ell^\infty(\Z^d)$ (endowed with its
natural norm). Consider a bounded linear mapping
$M:\;\ell^\infty(\Z^d)\to\ell^\infty(\Z^d)$ that admits a bounded
inverse.  Consider now $\omega=(\omega_n)_{n\in\Z^d}$ independent
random variables that are all bounded the same constant. Define the
random variables
\begin{equation}
  \label{eq:9}
  \tilde\omega=(\tilde\omega_n)_{n\in\Z^d}=M(\omega_n)_{n\in\Z^d}=M \omega.
\end{equation}
Fix $n_0$ and $m_0$ in $\Z^d$ arbitrary and write
\begin{equation}
  \label{eq:8}
  \begin{pmatrix}
    \tilde\omega_{m_0}\\\tilde\omega^{m_0}
  \end{pmatrix}
  =
  \begin{pmatrix}
    a & B\\ A & C
  \end{pmatrix}
  \begin{pmatrix}
    \omega_{n_0}\\\omega^{n_0}
  \end{pmatrix}
\end{equation}
where $\tilde\omega^{m_0}=(\tilde\omega_m)_{m\not=m_0}$ and
$\omega^{m_0}=(\omega_m)_{m\not=m_0}$. So, we consider $M$ as a map
from $\ell^2(\Z^d)=\C\delta_{n_0}\oplus\ell^2(\Z^d\setminus\{n_0\})$
to
$\ell^2(\Z^d)=\C\delta_{m_0}\oplus\ell^2(\Z^d\setminus\{m_0\})$. \\
Assume that
\begin{description}
\item[(I)] $a\not=0$ and $C$ is one-to-one.
\end{description}
Let $S_m$ be the concentration function of the random variable
$\omega_m$ and $\tilde S_m$ be the concentration function of the
random variable $\tilde \omega_m$ conditioned on
$(\tilde\omega_n)_{n\not=m}$. Then, one has
\begin{Le}
  \label{le:1}
  Under assumption (I), there exists $\kappa>0$ such that, for
  $x\geq0$ and for any bounded sequence $(\tilde\omega_n)_{n\not=m}$,
  one has $\tilde S_{m_0}(x)=S_{n_0}(\kappa\,x)$.
\end{Le}
\begin{proof}
  As $M$ is invertible with bounded inverse, one can write
  \begin{equation*}
    M^{-1}=
  \begin{pmatrix}\tilde a & \tilde B\\ \tilde A & \tilde
    C\end{pmatrix}
  \end{equation*}
  considering $M^{-1}$ as a map from
  $\ell^2(\Z^d)=\C\delta_{m_0}\oplus\ell^2(\Z^d\setminus\{m_0\})$ to
  $\ell^2(\Z^d)=\C\delta_{n_0}\oplus\ell^2(\Z^d\setminus\{n_0\})$. We,
  thus, obtain that
  \begin{equation}
    \label{eq:10}
    \begin{pmatrix}\tilde a & \tilde B\\ \tilde A & \tilde
      C\end{pmatrix} \begin{pmatrix}a & B\\ A & C\end{pmatrix}
    =\begin{pmatrix} 1 & 0\\ 0 & \text{Id} \end{pmatrix}=
    \begin{pmatrix} a & B\\ A & C \end{pmatrix}
    \begin{pmatrix}\tilde a & \tilde B\\ \tilde A & \tilde
      C\end{pmatrix}.
  \end{equation}
  Solving~\eqref{eq:8} in $\tilde\omega^{m_0}$ and $\omega_{n_0}$, we
  obtain $\tilde\omega_{m_0}=a\omega_{n_0}+B\omega^{n_0}$ and
  $\tilde\omega^{m_0}-\omega_{n_0}A=C\omega^{n_0}$. Thus, by the first
  equality in~\eqref{eq:10}, we have
  \begin{equation}
    \label{eq:11}
    \tilde C\tilde\omega^{m_0}-\omega_{n_0}\tilde
    CA=\omega^{n_0}-(B\omega^{n_0})\tilde A.    
  \end{equation}
  Applying $B$ to this, we get
  \begin{equation}
    \label{eq:12}
    (1-B\tilde A)B\omega^{n_0}=B\tilde
    C\tilde\omega^{m_0}-\omega_{n_0}B\tilde CA 
  \end{equation}
  We claim that, as $C$ is one-to-one, one has $B\tilde
  A\not=1$. Indeed, if $B\tilde A=1$, then $\tilde A\not=0$ and, by
  the first equality in~\eqref{eq:10}, ker$\,C=$\,span\,$\tilde
  A$. This
  contradicts our assumption on $C$.\\
  Hence, from~\eqref{eq:11},~\eqref{eq:12} and~\eqref{eq:10}, we
  obtain
  \begin{equation*}
    \tilde\omega_{m_0}=\frac{a}{1-B\tilde A}\,\omega_{n_0}-D\,\tilde\omega^{m_0}
  \end{equation*}
  where $D$ is a bounded linear form on
  $\ell^\infty(\Z^d\setminus\{m_0\})$.\\
  This immediately yields Lemma~\ref{le:1} if one sets $\D\kappa=
  \left|\frac{1-B\tilde A}{a}\right|$.
\end{proof}
\noindent Keeping the notations of the proof of Lemma~\ref{le:1}, we
easily checks that the assumption (I) is equivalent to the assumption
\begin{description}
\item[(I')] $a\cdot\tilde a \not=0$.
\end{description}
One can apply this to convolutions i.e. assume that $M$ is a
convolution, that is, it is given by a matrix $M=((\hat
M_{m-n}))_{(m,n)\in\Z^d\times\Z^d}$ (where $\hat M_n\in\R$). Assume
that
\begin{enumerate}
\item[(H1)] $\D\sum_{n\in\Z^d}|\hat M_n|<+\infty$,
\item[(H2)] the function $\D \theta\mapsto
  M(\theta)=\sum_{n\in\Z^d}\hat M_ne^{in\theta}$ does not vanish on
  $\R^d$.
\end{enumerate}
Consider now $(\omega_n)_{n\in\Z^d}$ independent random variables that
are all bounded by the same constant and define the random variables
$(\tilde\omega_n)_{n\in\Z^d}$ by~\eqref{eq:9}. Then, as a consequence
of Lemma~\ref{le:1}, we prove
\begin{Th}
  \label{thr:2}
  Assume (H1) and (H2). There exists $\kappa>0$ and $k\in\Z^d$ such
  that, for $n\in\Z^d$ and $x\geq0$, for any $(\tilde\omega_m)_{m\not
    n}$, one has $\tilde S_n(x)=S_{n+k}(\kappa\,x)$.
\end{Th}
\begin{proof}
  The fact that, under assumptions (H1) and (H2), $M$ is invertible on
  $\ell^\infty(\Z^d)$ with bounded inverse is an immediate consequence
  of Wiener's $1/f$-theorem (see e.g.~\cite[Chapter 2.4]{MR2359869}).\\
  Let $(\hat M^{-1}_n)_{n\in\Z^d}$ denote the Fourier coefficients of
  the function $\theta\mapsto M^{-1}(\theta)$. As $M\cdot M^{-1}=1$,
  one has $\D\sum_{n\in\Z^d}\hat M_n\hat M^{-1}_{-n}=1$. Thus, we pick
  $k$ such that
  \begin{equation}
    \label{eq:32}
    \hat M^{-1}_{-k}\hat M_{k}\not=0.  
  \end{equation}
  Fix $n_0\in\Z^d$ and write the decomposition~\eqref{eq:8} for
  $m_0=n_0+k$ and $n_0$. Then, the coefficient $a$, the vector $A$,
  the linear form $B$ and the operator $C$ in~\eqref{eq:8} are given
  by
  \begin{equation*}
    a=\hat M_k,\ A=(\hat M_{k-n})_{n\not=0},\  B=(\hat
    M_{k+m})_{m\not=0}\ \text{and}\ C=((\hat
    M_{k+m-n}))_{\substack{m\not=0\\n\not=0}}.
  \end{equation*}
  Note that $a$, $A$, $B$ and $C$ do not depend on $n_0$.\\
  In the same way, using the notations of the proof of
  Lemma~\ref{le:1}, one computes
  \begin{equation*}
    \tilde a=\hat M^{-1}_{-k},\ \tilde A=(\hat
    M^{-1}_{-k-n})_{n\not=0},\
    \tilde B=(\hat
    M^{-1}_{-k+m})_{m\not=0}\ \text{and}\ \tilde C=((\hat
    M^{-1}_{-k+m-n}))_{\substack{m\not=0\\n\not=0}}.
  \end{equation*} 
  Note that $\tilde a$, $\tilde A$, $\tilde B$ and $\tilde C$ do not
  depend on $n_0$.\\
  By construction (see~\eqref{eq:32}), we have $a\cdot\tilde a\not=0$;
  hence, assumption (I'), thus, assumption (I) is satisfied and the
  statement of Theorem~\ref{thr:2} for $\tilde S_{n_0}$ and
  $S_{n_0+k}$ follows immediately from Lemma~\ref{le:1}. Recalling
  that $a$, $A$, $B$, $C$, $\tilde a$, $\tilde A$, $\tilde B$ and
  $\tilde C$ do not depend on $n_0$, we obtain the full statement of
  Theorem~\ref{thr:2}.
\end{proof}
\noindent To show that the long range correlation model~\eqref{eq:13}
satisfies the assumption (\~ R) (thus, that the conclusions of
Theorems~\ref{thr:1},~\ref{thr:5},~\ref{thr:6} and~\ref{thr:7} and
Lemma~\ref{lemasympt} hold), we first note that $H_\omega$ defined
in~\eqref{eq:13} can be rewritten as $H_{\tilde\omega}$ defined
in~\eqref{eq:1} where $\tilde\omega=M \omega$ and $M$ is the
convolution associated to the multiplier $\D\theta\mapsto
M(\theta)=\sum_{n\in\Z^d}u_ne^{in\theta}$. Thus, the summability of
$u$ and the assumption (H) guarantee that assumptions (H1) and (H2)
are satisfied. That for the model~\eqref{eq:13} assumption (\~ R) is
satisfied is then an immediate consequence of Theorem~\ref{thr:2} and
assumption (R).
\subsection{The proof of Theorem~\ref{thr:4}}
\label{sec:proof-theorem}
We are now in a position to prove Theorem~\ref{thr:4}. To simplify
notations, we assume that the random variables $(\omega_n)_n$ are
centered; this comes up to shifting the operator $H_\omega$ by a
constant.
\par Recall that $\xi_j$ is defined in~\eqref{eq:15}. For $p>0$
arbitrary, for arbitrary, non empty, open, two by two disjoint intervals
$I_1,\dots,I_p$ and arbitrary integers $k_1,\cdots,k_p$, consider the
event
\begin{equation*}
  \Omega^\Lambda_{I_1,k_1;I_2,k_2;\cdots;I_p,k_p}:=\bigcap_{l=1}^p
  \left\{\omega;\
    \#\{j;\ \xi_j(\tilde\omega,\Lambda)\in I_l\}=k_l \right\}.
\end{equation*}
In view of Theorem~\ref{thr:7}, as
$\D\pro(\mathcal{Z}_\Lambda)\vers_{|\Lambda|\to+\infty}1$,
Theorem~\ref{thr:4} will be proved if we prove that
\begin{equation}
  \label{eq:110}
  \lim_{|\Lambda|\to+\infty}
  \pro\left(\Omega^\Lambda_{I_1,k_1;I_2,k_2;\cdots;I_p,k_p}
    \cap\mathcal{Z}_\Lambda\right)=
  \frac{|I_1|^{k_1}}{k_1!}e^{-|I_1|}\cdots
  \frac{|I_p|^{k_p}}{k_p!} e^{-|I_p|}.
\end{equation}
Pick $\varepsilon$ small, in particular, smaller than 
\begin{itemize}
\item the length of the smallest of the intervals $I_1,\dots,I_p$,
\item the smallest distance between two distinct intervals among
  $I_1,\dots,I_p$.
\end{itemize}
Define
\begin{equation*}
  I_j^{+,\varepsilon}=I_j\cup[-\varepsilon/2,\varepsilon/2]\quad\text{ and
  }\quad I_j^{-,\varepsilon}=I_j\cap(^cI_j+[-\varepsilon/2,\varepsilon/2]). 
\end{equation*}
Clearly, $I_j^{-,\varepsilon}\subset I_j\subset I_j^{+,\varepsilon}$.\\
For a cube $\Lambda$ and an interval $I$, define the Bernoulli random
variable $X_{\Lambda_\ell,\ell',I}$ by
\begin{equation*}
  X_{\Lambda_\ell,\ell',I}=\car_{H_\omega(\Lambda_\ell)\text{ has an
      e.v. in }E_0+|\Lambda|n(E_0)]^{-1}\,I\text{ with localization
      center in }\Lambda_{\ell-\ell'}}.
\end{equation*}
Here, the length scales $\ell$ and $\ell'$ are taken as in
Theorem~\ref{thr:7} that is $\ell\asymp L^\beta$ and $\ell'\asymp
L^{\beta'}$. Notice that, using the notations of
section~\ref{sec:local-distr-eigenv}, one has
$X_{\Lambda_\ell,\ell',I}= X(\Lambda_\ell,N^{-1}[N(E_0)+
|\Lambda|^{-1}J],\ell')$ where, as it is assumed that $N$ is
differentiable at $E_0$ and that $n(E_0)=N'(E_0)>0$, one has $|J\Delta
I|=o(1)$ as $\Lambda\to+\infty$; here, $J\Delta I$ denotes the
symmetric difference between $I$ and $J$ i.e.  $J\Delta I=(I\setminus
J)\cup(J\setminus I)$.\\
Consider now the decomposition in cubes
$(\Lambda_\ell(\gamma_j))_{1\leq j\leq J}$ given by
Theorem~\ref{thr:7}; then, $J=L^{d(1-\beta)}(1+o(1))$. The choice of
the parameters $\tilde\rho$, $\alpha$, $\beta$ and $\beta'$ is the
following. The parameter $\beta'$ will be the determining parameter
and will be chosen close to $1$ (see Lemma~\ref{le:3} below); thus,
the parameters $\alpha$ and $\beta$ will also be close to $1$ and
$\tilde\rho$ will be chosen close to $0$; indeed, in
Theorem~\ref{thr:7}, one requires
\begin{equation*}
  0<\beta'<\beta,\quad0<\alpha<1,\quad0<\tilde\rho\quad\text{and}\quad
  1+\beta<\frac{2\alpha}{1+\tilde\rho}.
\end{equation*}
By Theorem~\ref{thr:7}, in particular~\eqref{error}, and the Wegner
estimate~\eqref{eq:7}, we know that, for $L$ sufficiently large, one
has
\begin{equation*}
  \begin{aligned}
    \bigcap_{l=1}^p \left\{\omega;\ \sum_{j=1}^J
      X_{\Lambda_\ell(\gamma_j),\ell',I^{-,\varepsilon}_l}=k_l
    \right\}\bigcap\mathcal{Z}_\Lambda\subset
    \Omega^\Lambda_{I_1,k_1;I_2,k_2;\cdots;I_p,k_p}\cap\mathcal{Z}_\Lambda,\\
    \Omega^\Lambda_{I_1,k_1;I_2,k_2;\cdots;I_p,k_p}\cap\mathcal{Z}_\Lambda
    \subset\bigcap_{l=1}^p \left\{\omega;\ \sum_{j=1}^J
      X_{\Lambda_\ell(\gamma_j),\ell',I^{+,\varepsilon}_l}=k_l
    \right\}\bigcap\mathcal{Z}_\Lambda.
  \end{aligned}
\end{equation*}
Hence, as $\D\pro(\mathcal{Z}_\Lambda)\vers_{|\Lambda|\to+\infty}1$,
it suffices to prove that, for any $\delta>0$, there exists
$\varepsilon>0$ small such that
\begin{equation}
  \label{eq:20}
  \begin{aligned}
    \liminf_{|\Lambda|\to+\infty}\pro \left(\bigcap_{l=1}^p
      \left\{\sum_{j=1}^J
        X_{\Lambda_\ell(\gamma_j),\ell',I^{-,\varepsilon}_l}=k_l
      \right\}\right)\geq\prod_{j=1}^p
    \frac{|I_j|^{k_j}}{k_j!}e^{-|I_j|}-\delta,
    \\
    \limsup_{|\Lambda|\to+\infty}\pro \left(\bigcap_{l=1}^p
      \left\{\sum_{j=1}^J
        X_{\Lambda_\ell(\gamma_j),\ell',I^{+,\varepsilon}_l}=k_l
      \right\}\right)\leq\prod_{j=1}^p
    \frac{|I_j|^{k_j}}{k_j!}e^{-|I_j|}+\delta.
  \end{aligned}
\end{equation}
The main loss compared to the case when the (IAD) assumption holds is
that the random variables
$(X_{\Lambda_\ell(\gamma_j),\ell',I^{\pm,\varepsilon}_l})_{j}$ are not
independent anymore. We will show that this dependence can be
controlled using assumption (D) and the decorrelation estimates
obtained in Theorem~\ref{thr:1}.\\
Now, for $\ell''\leq\ell'\leq\ell$ and $\gamma\in\Z^d$, define the
auxiliary operator $\tilde H_\omega(\Lambda_\ell(\gamma),\ell'')$ to
be the operator
\begin{equation*}
  \tilde
  H_{\omega,\ell''}:=-\Delta+\lambda\sum_{n\in\Z^d}\tilde\omega_n\tau_n
  (u_{\ell''})  \text{ where }u_{\ell''}(m)=
  \begin{cases}
    u(m)&\text{ if }|m|\leq \ell'',\\0&\text{ if not}
  \end{cases}
\end{equation*}
restricted to $\Lambda_\ell(\gamma)$ (with periodic boundary
conditions).\\
We prove
\begin{Le}
  \label{le:4}
  There exists $\ell_0$ such that if $\ell''\geq\ell_0$, then, for
  $\lambda$ sufficiently large, the whole spectrum of $\tilde
  H_{\omega,\ell''}$ is localized, in particular, the conclusions of
  Theorems~\ref{thr:5} and~\ref{thr:6} hold for this operator for a
  value $\eta$ independent of $\ell''$. Moreover, the spectral
  estimates given in Theorem~\ref{thr:1} also hold for $\tilde
  H_{\omega,\ell''}$ with constants independent of $\ell''$.
\end{Le}
\begin{proof}[Proof of Lemma~\ref{le:4}]
  Clearly, to prove Lemma~\ref{le:4}, it suffices to 
  \begin{enumerate}
  \item prove that if $u$ satisfies (S) and (H), there exists $\ell_0$
    such that, for $\ell''\geq\ell_0$, the potential $u_{\ell''}$
    satisfies (S) and (H) uniformly in $\ell''$,
  \item reapply the arguments explained above for $H_{\omega}$.
  \end{enumerate}
  Assume that
  \begin{equation*}
    m:=\min_{\theta\in\R^d}\left|\sum_{n\in\Z^d}u_ne^{in\theta}\right|>0.
  \end{equation*}
  Thus, by (S), we can pick $\ell_0$ such that
  $\D\sum_{|n|\geq\ell_0}|u(n)|\leq m/2$; then, for
  $\ell''\geq\ell_0$, we have that
  \begin{equation*}
    \min_{\theta\in\R^d}\left|\sum_{n\in\Z^d}u^{\ell''}_ne^{in\theta}\right|
    \geq\min_{\theta\in\R^d}\left|\sum_{n\in\Z^d}u_ne^{in\theta}\right|-m/2
    \geq m/2.
  \end{equation*}
  This completes the proof of Lemma~\ref{le:4}.
\end{proof}
\noindent We also define the Bernoulli random variable
$X_{\Lambda_\ell,\ell',\ell'',I}$ by
\begin{equation*}
  X_{\Lambda_\ell,\ell',\ell'',I}=\car_{\tilde
    H_{\omega,\ell''}(\Lambda_\ell)\text{ has an 
      e.v. in }E_0+|\Lambda|n(E_0)]^{-1}\,I\text{ with localization
      center in }\Lambda_{\ell-\ell'}}.
\end{equation*}
Then, clearly, by Theorem~\ref{thr:7}, one has that
\begin{Le}
  \label{le:2}
  Assume $\ell''\leq\ell/3$. For $I$ a real interval, the random
  variables $(X_{\Lambda_\ell(\gamma_j),\ell',\ell'',I})_{1\leq j\leq
    J}$ are two by two independent.
\end{Le}
\noindent We also prove
\begin{Le}
  \label{le:3}
  Assume (D) holds. Recall that $\eta>d-1/2$ and assume that
  $\D\beta'\in\left(\frac d{2\eta-d+1},1\right)$. Set
  $\delta:=-d+(2\eta-d+1)\beta'>0$.\\
  Then, for any $a<b$, any $\varepsilon\in(0,(b-a)/2)$ and any $p>0$,
  for $L$ sufficiently large, with probability at least $1-L^{-p}$,
  for all $1\leq j\leq J$, one has
  \begin{equation}
    \label{eq:18}
    X_{\Lambda_\ell(\gamma_j),2\ell'/3,\ell'/3,(a+\varepsilon,b-\varepsilon)}\leq 
    X_{\Lambda_\ell(\gamma_j),\ell',(a,b)}\leq
    X_{\Lambda_\ell(\gamma_j),4\ell'/3,\ell'/3,(a-\varepsilon,\beta+\varepsilon)}. 
  \end{equation}
\end{Le}
\begin{proof}[Proof of Lemma~\ref{le:3}]
  By Hoeffding's inequality (see e.g.~\cite{MR2319879}), we know that,
  for some $C>0$ (depending only on the essential supremum of the
  random variables $(|\omega_n|)_n$), for $\ell'=L^{\beta'}/3$, for
  $L$ sufficiently large, one has
  \begin{equation}
    \label{eq:21}
    \begin{split}
      \pro\left(\left|\sum_{|m|\geq\ell'}\omega_{n-m} u(m) \right|\geq
        \varepsilon L^{-d}\right)&\leq \exp \left(-\frac{\varepsilon
          L^{-d}}{\D C\,\sum_{|m|\geq\ell'}|u(m)|^2} \right)\\&\leq
      e^{-\varepsilon L^{\delta}}
    \end{split}
  \end{equation}
  where we recall that $\delta=-d+(2\eta-d+1)\beta'>0$; here, we have
  used assumption (D), the fact that $\eta>d-1/2$ and picked $\D \frac
  d{2\eta-d+1}<\beta'<1$.\\
  Hence, with a probability at least $1-L^d e^{-\varepsilon
    L^{\delta}}$, we have that
  \begin{equation*}
    \sup_{1\leq j\leq J}\sup_{m\in\Lambda_\ell(\gamma_j)}
    \left|\left(\sum_{n\in\Z^d}\omega_n \tau_n u\right)(m)-\tilde\omega_m
    \right|\leq \varepsilon L^{-d}.
  \end{equation*}
  that is
  \begin{equation}
    \label{eq:22}
    \sup_{1\leq j\leq J}
    \left\|\tilde H_{\omega,,\ell''}(\Lambda_\ell(\gamma_j))
      -H_\omega(\Lambda_\ell(\gamma_j))
    \right\|\leq \varepsilon L^{-d}.
  \end{equation}
  Next, we prove
  \begin{Le}
    \label{le:5}
    Assume that the conclusions of Theorem~\ref{thr:6} hold for the
    random operator $H_\omega$ for all $L$ sufficiently large. Fix
    $p>0$ and $r>0$. Let $q$ and $\eta$ be given by
    Theorem~\ref{thr:6}.\\
    Then, for $L$ sufficiently large, with probability at least
    $1-L^{-p}$, for $L'\leq L$ and $\gamma\in\Lambda_L$ such that
    $\Lambda_{L'+\eta^{-1}(q+r+d/2)\log L}(\gamma)\subset\Lambda_L$
    ($q$ being defined in Theorem~\ref{thr:6}), if there exists
    $\varphi\in\ell^2(\Lambda_L)$ such that
    \begin{itemize}
    \item     supp$\varphi\subset\Lambda_{L'}(\gamma)$,
    \item $\|\varphi\|=1$ and
      $\|(H_\omega(\Lambda_L)-E)\varphi\|\leq\varepsilon L^{-d}$,
    \end{itemize}
    then, $H_\omega(\Lambda_L)$ has an eigenvalue in $[E-2\varepsilon
    L^{-d},E+2\varepsilon L^{-d}]$ with localization center in
    $\Lambda_{L'+\eta^{-1}(q+r+d/2)\log L}(\gamma)$.
  \end{Le}
  \begin{proof}[Proof of Lemma~\ref{le:5}]
    Pick $\varphi$ as in Lemma~\ref{le:5}. As $H_\omega(\Lambda_L)$ is
    self-adjoint and $\|(H_\omega(\Lambda_L)-E)\varphi\|\leq\varepsilon$,
    $H_\omega(\Lambda_L)$ has an eigenvalue in
    $[E-\varepsilon,E+\varepsilon]$. Expand $\varphi$ in the basis of
    eigenfunctions of $H_\omega(\Lambda_L)$:
    $\varphi=\sum_{i}\langle\varphi,\varphi_i\rangle\varphi_i$. If the
    localization centers of $\varphi_i$ are outside of
    $\Lambda_{L'+\eta^{-1}(q+r+d/2)\log L}(\gamma)$, as
    supp$\varphi\subset\Lambda_{L'}(\gamma)$, one has
    $|\langle\varphi,\varphi_i\rangle|\lesssim L^{-r-d/2}$. On the
    other hand,
    \begin{equation}
      \label{eq:23}
      \varepsilon^2L^{-2d}\geq\|(H-E)\varphi\|^2=
      \sum_{i}|\langle\varphi,\varphi_i\rangle|^2|E-E_i|^2.
    \end{equation}
    Thus, if the conclusion of Lemma~\ref{le:5} does not hold, then,
    as the total number of eigenvalues is bounded by $L^d$, we have
    \begin{equation*}
      1-\sum_{|E-E_i|>2\varepsilon L^{-d}}
      |\langle\varphi,\varphi_i\rangle|^2=\sum_{|E-E_i|\leq2\varepsilon
        L^{-d}}|\langle\varphi,\varphi_i\rangle|^2\lesssim L^{-2r},
    \end{equation*}
    thus, by~\eqref{eq:23},
    \begin{equation*}
      \varepsilon^2L^{-2d}\geq\|(H-E)\varphi\|^2=
      \sum_{i}|\langle\varphi,\varphi_i\rangle|^2|E-E_i|^2\geq
      4\varepsilon^2 L^{-2d}\left(1-O(L^{-2r})\right).
    \end{equation*}
    which is absurd for $L$ sufficiently large.\\
    Lemma~\ref{le:5} is proved.
  \end{proof}
  \noindent Now, we can use~\eqref{eq:22} and apply Lemma~\ref{le:5}
  in turn to $H_\omega(\Lambda_\ell(\gamma_j),\ell')$ and to
  $H_\omega(\Lambda_\ell(\gamma_j))$ to obtain that, for $L$
  sufficiently large, with probability at least $1-L^{-p}$ ($p>0$
  fixed arbitrary),
  \begin{itemize}
  \item if $H_\omega(\Lambda_\ell(\gamma_j))$ has an eigenvalue in
    $(a,b)$ with loc. center in $\Lambda_{\ell-\ell'}(\gamma_j)$ then
    $H_{\omega,,\ell''}(\Lambda_\ell(\gamma_j))$ has an eigenvalue in
    $(a-\varepsilon L^{-d},b+\varepsilon L^{-d})$ with loc. center in
    $\Lambda_{\ell-\ell'+\ell''}(\gamma_j)$;
  \item if $H_{\omega,,\ell''}(\Lambda_\ell(\gamma_j))$ has an eigenvalue
    in $(a,b)$ with loc. center in $\Lambda_{\ell-\ell'}(\gamma_j)$
    then $H_\omega(\Lambda_\ell(\gamma_j))$ has an eigenvalue in
    $(a-\varepsilon L^{-d},b+\varepsilon L^{-d})$ with loc. center in
    $\Lambda_{\ell-\ell'+\ell''}(\gamma_j)$.
  \end{itemize}
  This, then, implies~\eqref{eq:18} and completes the proof of
  Lemma~\ref{le:3}.
\end{proof}
\noindent Equipped with Lemma~\ref{le:3}, in view of~\eqref{eq:20}, to
prove Theorem~\ref{thr:4}, it suffices to prove that, for any $a>0$
and $b>0$, for any $\delta\in(0,1/3]$, there exists $\varepsilon>0$ small such
that
\begin{equation}
  \label{eq:29}
  \begin{aligned}
    \liminf_{|\Lambda|\to+\infty}\pro \left(\bigcap_{l=1}^p
      \left\{\sum_{j=1}^J X_{\Lambda_\ell(\gamma_j),a\ell',b\ell',I^{-,\varepsilon}_l}=k_l
      \right\}\right)\geq\prod_{j=1}^p
    \frac{|I_j|^{k_j}}{k_j!}e^{-|I_j|}-\delta,
    \\
    \limsup_{|\Lambda|\to+\infty}\pro \left(\bigcap_{l=1}^p
      \left\{\sum_{j=1}^J X_{\Lambda_\ell(\gamma_j),a\ell',b\ell',I^{+,\varepsilon}_l}=k_l
      \right\}\right)\leq\prod_{j=1}^p
    \frac{|I_j|^{k_j}}{k_j!}e^{-|I_j|}+\delta.
  \end{aligned}
\end{equation}
We will only prove the second inequality, the first one being proved
in the same way. \\
First, by~\eqref{eq:51} of Lemma~\ref{lemasympt}, by Lemma~\ref{le:3}
and by the definition of $(I_l^{\pm,\varepsilon})_l$, for $L$
sufficiently large and our choice of $\varepsilon$, one has
\begin{equation}
  \label{eq:108}
  \begin{split}
    \pro(X_{\Lambda_\ell(\gamma_j),a\ell',b\ell',I^{\pm,\varepsilon}_l}=1)
    &=(|I_l|\pm\varepsilon)L^{d(\beta-1)}(1+o(1))
    \\&=(|I_l|\pm\varepsilon)J^{-1}(1+o(1)).
  \end{split}
\end{equation}
Then, we compute
\begin{equation}
  \label{eq:25}
  \begin{split}
    &\pro \left(\bigcap_{l=1}^p \left\{\sum_{j=1}^J
        X_{\Lambda_\ell(\gamma_j),a\ell',b\ell',I^{+,\varepsilon}_l}=k_l
      \right\}\right)
    \\&= \sum_{\substack{ K_l\subset\{1,\cdots,\tilde
        N\}\\ \#K_l=k_l,\ 1\leq l\leq p}}\pro \left(\bigcap_{l=1}^p
        \left\{\omega;
        \begin{aligned}
          \forall j\in K_l,\ X_{\Lambda_\ell(\gamma_j),a\ell',b\ell',
            I^{+,\varepsilon}_l}&=1\\ \forall j\not\in K_l,\
          X_{\Lambda_\ell(\gamma_j),a\ell',b\ell',I^{+,\varepsilon}_l}&=0
        \end{aligned}
      \right\} \right).
  \end{split}
\end{equation}
For $1\leq l\leq p$, pick $K_l\subset\{1,\cdots,J\}$ such that
$(\#K_l)_{1\leq l\leq p}=(k_l)_{1\leq l\leq p}$. For $1\leq j\leq J$,
define $\D \kappa_j=\#\{1\leq l\leq p;\ \gamma_j\in K_l\}
=\sum_{l=1}^p\car_{\gamma_j\in K_l}$. Then, one has
\begin{equation}
  \label{eq:24}
  \sum_{j=1}^J\kappa_j=\sum_{l=1}^p\sum_{j=1}^J\car_{\gamma_j\in K_l}
  =k_1+\cdots+k_p. 
\end{equation}
As $\cup_l I_l^{+,\varepsilon}\subset I_C:=[-C,C]$ (for some $C>0$),
thanks to the decorrelation estimates~\eqref{eq:6} for
$H_{\omega,,\ell''}(\Lambda_\ell(\gamma_j))$ (see Lemma~\ref{le:4}) and
using their stochastic independence, one has the following a priori
bound
\begin{equation}
  \label{eq:21}
  \begin{split}
    &\pro \left(\bigcap_{l=1}^p
        \left\{\omega;
        \begin{aligned}
          \forall j\in K_l,\ X_{\Lambda_\ell(\gamma_j),a\ell',b\ell',I^{+,\varepsilon}_l}&=1\\
          \forall j\not\in K_l,\ X_{\Lambda_\ell(\gamma_j),a\ell',b\ell',I^{+,\varepsilon}_l}&=0
        \end{aligned}
      \right\} \right)\\&\leq \prod_{j=1}^J \pro \left\{\omega;
      H_{\omega,,\ell''}(\Lambda_\ell(\gamma_j))\text{ has at least
      }\kappa_j\text{ e.v. in }I_C \right\}\\
    &\leq \prod_{j=1}^J
    (CL^{-d}|\Lambda_\ell(\gamma_j)|)^{\kappa_j}\leq C
    J^{-k_1-k_2-\cdots-k_p}
  \end{split}
\end{equation}
as
$J=L^{d}|\Lambda_\ell(\gamma_j)|^{-1}(1+o(1))=L^{d(1-\beta)}(1+o(1))$;
here, we have used~\eqref{eq:24}.\\
By~\eqref{eq:25}, we have
\begin{equation}
  \label{eq:26}
  \begin{split}
    &\pro \left(\bigcap_{l=1}^p \left\{\sum_{j=1}^J
        X_{\Lambda_\ell(\gamma_j),a\ell',b\ell',I^{+,\varepsilon}_l}=k_l
      \right\}\right) \\&= \sum_{\substack{
        K_l\subset\{1,\cdots,\tilde N\}\\ \#K_l=k_l,\ 1\leq l\leq
        p\\\forall l\not= l',\ K_l\cap K_{l'}=\emptyset }}\pro
    \left(\bigcap_{l=1}^p \left\{\omega;
        \begin{aligned}
          \forall j\in K_l,\ X_{\Lambda_\ell(\gamma_j),a\ell',b\ell',
            I^{+,\varepsilon}_l}&=1\\ \forall j\not\in K_l,\
          X_{\Lambda_\ell(\gamma_j),a\ell',b\ell',I^{+,\varepsilon}_l}&=0
        \end{aligned}
      \right\} \right) \\&\hskip.5cm+ \sum_{\substack{
        K_l\subset\{1,\cdots,\tilde N\}\\ \#K_l=k_l,\ 1\leq l\leq p\\
        \exists l\not= l',\ K_l\cap K_{l'}\not=\emptyset}}\pro
    \left(\bigcap_{l=1}^p \left\{\omega;
        \begin{aligned}
          \forall j\in K_l,\ X_{\Lambda_\ell(\gamma_j),a\ell',b\ell',
            I^{+,\varepsilon}_l}&=1\\ \forall j\not\in K_l,\
          X_{\Lambda_\ell(\gamma_j),a\ell',b\ell',I^{+,\varepsilon}_l}&=0
        \end{aligned}
      \right\} \right).
  \end{split}
\end{equation}
One sets and, by Lemmas~\ref{lemasympt} and~\ref{le:3}, one computes
\begin{equation*}
  p_l^+:=\pro\left(X_{\Lambda_\ell(\gamma_j),a\ell',b\ell',
    I^{+,\varepsilon}_l}=1\right)=|I_l^{+,\varepsilon}|J^{-1}(1+o(1)).
\end{equation*}
We also will use~\cite[Lemma 4.1]{Ge-Kl:10} 
\begin{Le}
  \label{le:6}
  With the choice of $(I_l)_{1\leq l\leq p}$ made above, under the
  assumptions of Theorem~\ref{thr:4}, with our choice of $\ell$ and
  $\ell'$, for $L$ sufficiently large, one has
  \begin{equation*}
    \pro\left(\sum_{l=1}^pX_{\Lambda_\ell(\gamma_j),a\ell',b\ell',
        I^{\pm,\varepsilon}_l}=0\right)
    = 1-(1-o(1))J^{-1}\sum_{l=1}^p|I_l^{\pm,\varepsilon}|. 
  \end{equation*}
\end{Le}
\noindent For $1\leq l\leq p$, pick $K_l\subset\{1,\cdots,J\}$ such
that $(\#K_l)_{1\leq l\leq p}=(k_l)_{1\leq l\leq p}$ and $K_l\cap
K_{l'}=\emptyset$ if $l\not= l'$. For such $(K_l)_{1\leq l\leq p}$,
using the stochastic independence of the operators
$(H_{\omega,,\ell''}(\Lambda_\ell(\gamma_j)))_{1\leq j\leq J}$, one
computes
\begin{equation}
  \label{eq:27}
  \begin{split}
    &\pro \left(\bigcap_{l=1}^p
        \left\{\omega;
        \begin{aligned}
          \forall j\in K_l,\ X_{\Lambda_\ell(\gamma_j),a\ell',b\ell',
            I^{+,\varepsilon}_l}&=1\\ \forall j\not\in K_l,\
          X_{\Lambda_\ell(\gamma_j),a\ell',b\ell',I^{+,\varepsilon}_l}&=0
        \end{aligned}
      \right\} \right)\\&=  \pro\left(\bigcap_{l=1}^p \left\{\omega;
      \begin{aligned}
        \forall j\in K_l,\ X_{\Lambda_\ell(\gamma_j),a\ell',b\ell',
            I^{+,\varepsilon}_l}&=1\\
        \forall j\not\in\bigcup_{l'\not=l}K_{l'},\
        X_{\Lambda_\ell(\gamma_j),a\ell',b\ell',
            I^{+,\varepsilon}_l}&=0
      \end{aligned}
    \right\} \right)\\&= \prod_{j\not\in\cup_{l=1}^p K_l}
  \pro\left(\sum_{l=1}^pX_{\Lambda_\ell(\gamma_j),a\ell',b\ell',
      I^{+,\varepsilon}_l}=0\right)\\&\hskip3cm\cdot
  \prod_{l=1}^p\prod_{j\in K_l} \pro\left(
    X_{\Lambda_\ell(\gamma_j),a\ell',b\ell',
      I^{+,\varepsilon}_l}=1\right) \\& =
  \left(1-\sum_{l=1}^pp^+_l\right)^{J-(k_1+\cdots+k_p)} \prod_{l=1}^p
  \left[p^+_l\right]^{k_l}(1+o(1)) \\& =
  \left(\prod_{l=1}^pe^{-|I^{+,\varepsilon}_l|}|I^{+,\varepsilon}_l|^{k_l}
  \right) J^{-k_1-\cdots-k_p}(1+o(1)).
  \end{split}
\end{equation}
On the other hand, as $k_1+\cdots+k_p$ is bounded, when $J\to+\infty$,
that is, when $|\Lambda|\to+\infty$, one has
\begin{equation}
  \label{eq:28}
  \sum_{\substack{
      K_l\subset\{1,\cdots,J\}\\ \#K_l=k_l,\ 1\leq l\leq
      p}}1=
  \prod_{l=1}^p\binom{J}{k_l}\quad\text{ and } 
  \sum_{\substack{
      K_l\subset\{1,\cdots,J\}\\ \#K_l=k_l,\ 1\leq l\leq
      p\\\exists l\not= l',\ K_l\cap K_{l'}\not=\emptyset}}1=o
  \left(  \prod_{l=1}^p\binom{J}{k_l}\right).
\end{equation}
Thus, for $L$ sufficiently large, plugging the a-priori
bound~\eqref{eq:21}, the bounds~\eqref{eq:28} and~\eqref{eq:27}
into~\eqref{eq:26}, we compute
\begin{equation*}
  \begin{split}
    &\pro \left(\bigcap_{l=1}^p \left\{\sum_{j=1}^J
        X_{\Lambda_\ell(\gamma_j),a\ell',b\ell',I^{+,\varepsilon}_l}=k_l
      \right\}\right) \\&=
    \left(\prod_{l=1}^pe^{-|I^{+,\varepsilon}_l|}|I^{+,\varepsilon}_l|^{k_l}
    \right) \prod_{l=1}^p\binom{J}{k_l} J^{-k_1-\cdots-k_p}+o
    \left(J^{-k_1-k_2-\cdots-k_p}
      \prod_{l=1}^p\binom{J}{k_l}\right)\\&=
    \prod_{l=1}^p\frac{e^{-|I^{+,\varepsilon}_l|}|I^{+,\varepsilon}_l|^{k_l}}{k_l!}
    +o(1)= \prod_{l=1}^p\frac{e^{-|I_l|}|I_l|^{k_l}}{k_l!}
    +O(\varepsilon).
  \end{split}
\end{equation*}
This proves the second inequality in~\eqref{eq:29}, thus,
in~\eqref{eq:20}. The first one is proved in the same way. The proof
of Theorem~\ref{thr:4} is complete.\qed
\def\cprime{$'$} \def\cydot{\leavevmode\raise.4ex\hbox{.}} \def\cprime{$'$}

\end{document}